\documentclass{l4dc2025}

\title[]{Anytime Safe Reinforcement Learning}

\usepackage{times}
\usepackage{graphicx}
\usepackage{amsfonts}
\usepackage{algorithm}
\usepackage{algorithmic}
\usepackage{tablists}
\usepackage{authblk}
\usepackage{enumitem}

\renewcommand\labelenumi{(\roman{enumi})}
\renewcommand\theenumi\labelenumi

\AtBeginEnvironment{algorithm2e}{\algoequations}
\AtEndEnvironment{algorithm2e}{\restoreequations}
\newcounter{algosavedequation}
\newcommand{\algoequations}{%
  \setcounter{algosavedequation}{\value{equation}+1}%
  \setcounter{equation}{0}%
  \renewcommand{\theequation}{\arabic{algosavedequation}\alph{equation}}
}
\newcommand{\restoreequations}{%
  \setcounter{equation}{\value{algosavedequation}}%
}

\allowdisplaybreaks[3]

\newcommand{\longthmtitle}[1]{\mbox{}\emph{(#1):}}

\newtheorem{assumption}{Assumption}
	
\newcommand{\real}{\mathbb{R}}

\newcommand\oprocendsymbol{\hbox{$\square$}}
\newcommand\oprocend{\relax\ifmmode\else\unskip\hfill%
\fi\oprocendsymbol}

\newcommand{\Var}{\mathrm{Var}}

\DeclareMathAlphabet{\mymathbb}{U}{BOONDOX-ds}{m}{n}

\newcommand{\setdef}[2]{\{#1 \; : \; #2\}}

\newcommand{\argmax}[2] {\mathrm{arg}\max_{#1}#2}
\newcommand{\argmin}[2] {\mathrm{arg}\min_{#1}#2}

\newcommand\xqed[1]{%
  \leavevmode\unskip\penalty9999 \hbox{}\nobreak\hfill
  \quad\hbox{#1}}
\newcommand\demo{\xqed{$\bullet$}}

\newcommand{\Cc}{{\mathcal{C}}}

\newcommand{\Lc}{{\mathcal{L}}}
\newcommand{\Sc}{{\mathcal{S}}}
\newcommand{\Ac}{{\mathcal{A}}}
\newcommand{\Ic}{{\mathcal{I}}}

\newcommand{\Fc}{{\mathcal{F}}}
\newcommand{\Oc}{{\mathcal{O}}}

\newcommand{\norm}[1]{\left\lVert#1\right\rVert}

\author[ ]{%
 \Name{Pol Mestres}$^1$ \Email{pomestre@ucsd.edu} \\
 \Name{Arnau Marzabal}$^{1,2}$ \Email{amarzabal@ucsd.edu} \\
 \Name{Jorge Cortés}$^1$ \Email{cortes@ucsd.edu}
}

\affil[ ]{$^1$Department of Mechanical and Aerospace Engineering, University of California San Diego}
\affil[ ]{$^2$Centre de Formació Interdisciplinària Superior, Universitat Politècnica de Catalunya}

\begin{document}

\maketitle

\begin{abstract}
  This paper considers the problem of solving constrained
  reinforcement learning problems with anytime guarantees, meaning
  that the algorithmic solution returns a safe policy regardless of
  when it is terminated.  Drawing inspiration from anytime constrained
  optimization, we introduce Reinforcement Learning-based Safe
  Gradient Flow (RL-SGF), an on-policy algorithm which employs
  estimates of the value functions and their respective gradients
  associated with the objective and safety constraints for the current
  policy, and updates the policy parameters by solving a convex
  quadratically constrained quadratic program.  We show that if the
  estimates are computed with a sufficiently large number of episodes
  (for which we provide an explicit bound), safe policies are updated
  to safe policies with a probability higher than a prescribed
  tolerance. We also show that iterates asymptotically converge to a
  neighborhood of a KKT point, whose size can be arbitrarily reduced
  by refining the estimates of the value function and their gradients.
  We illustrate the performance of RL-SGF in a navigation example.
\end{abstract}

\section{Introduction}

Reinforcement Learning (RL) seeks to find an optimal policy for an
agent that interacts in a given environment through a process of trial
and error. At every state, the agent chooses an action, after which it
randomly transitions to a new state and obtains the corresponding
reward.  The optimal policy is that which maximizes a predefined
long-horizon reward.  Formally, this process is modeled as a Markov
Decision Process (MDP) and has exhibited a lot of empirical success in
a variety of application domains.
However, in many safety-critical applications, such as autonomous
driving, robotic manipulation, or frequency control of power systems,
the process of trial and error executed by most RL algorithms can lead
the agent towards unsafe regions, with potentially catastrophic
consequences.  This observation has spurred a lot of interest in the
development of safe reinforcement learning algorithms, which seek to
find the optimal policy and respect safety constraints.

\smallskip \textsc{Literature Review:} The safe RL literature is vast
and, in what follows, we focus on works which are most aligned with
the present manuscript.  For more exhaustive surveys on safe RL, we
refer the reader
to~\citep{LB-MG-AWH-ZY-SZ-JP-APS:22,YL-AH-XL:21,JG-FF:15,SG-LY-YD-GC-FW-JW-AK:22}.
Safety constraints in RL are often expressed in the form of
\textit{cumulative constraints}, i.e., expected cumulative rewards
that need to be kept below a certain threshold over certain time
horizon~\citep{EA:99,JA-DH-AT-PA:17,YC-MG-LJ-MP:17,SP-LC-MCF-AR:19}.
MDPs with such type of constraints are referred to as
Constrained Markov Decision Processes (CMDPs).
A popular approach to solve CMDPs are primal-dual
methods~\citep{SP-MCF-LFOC-AR:23,DD-KZ-TB-MJ:20},
which simultaneously perform a maximization step in the primal
variable and a minimization step in the dual variable
and can be shown to converge to the optimal policy for finite state
and action CMDPs for a special class of probability transition
functions~\citep{DD-XW-ZY-ZW-MJ:21}.
In continuous state-action space,~\citep{SP-LC-MCF-AR:19} also
provides a primal-dual scheme, and shows that if a non-convex unconstrained RL problem 
is solved at every iteration, the algorithm provably converges to the optimal
policy.   
However, solving such unconstrained RL problem at every iteration is computationally intractable.
Primal-dual algorithms do not guarantee the satisfaction of the safety
constraints during the training process. Other works have employed
primal-dual methods to guarantee safety during training, but are
either limited to particular policy
parameterizations~\citep{SZ-TTD-JR:22} or solve a relaxed version of
the problem and hence introduce an optimality
gap~\citep{QB-ASB-MA-AK-VA:22,QB-ASB-VA:23}.
On the other hand,~\citep{JA-DH-AT-PA:17} proposes CPO, an algorithm
purely based on primal variable updates that enjoys safety guarantees
at every iteration.  However, since the exact policy update law is
computationally intensive, the work
provides a practical algorithm based on a first-order approximations of
the objective and constraints that might not satisfy the safety
constraints during training.
Alternatively,~\citep{YL-JD-XL:20} presents IPO, another primal method
that adds the constraints as penalty terms in the objective function,
and also guarantees the satisfaction of safety constraints during
training.  However, it requires a feasible initial policy and does not
possess formal convergence guarantees.
Other primal methods such as~\citep{YC-ON-ED-MG:18}
leverage Lyapunov functions to guarantee the satisfaction of
constraints during training.
However, the method proposed to search for such Lyapunov functions might be computationally
intensive, and convergence guarantees are only given for a
limited class of problems.
Finally,~\citep{WS-VKS-KCK-SS-JL-VG-BMS:24} optimizes over a class of
truncated policies so that unsafe actions have probability zero, but
the restriction to such class of policies also introduces an
optimality gap, which is not formally quantified.

\smallskip \textsc{Statement of Contributions:} We consider the
problem of designing an algorithm that finds the optimal policy of a
constrained RL problem within a parametric family while satisfying the
safety constraints at every iterate. Inspired by recent advances in
anytime constrained optimization, we introduce Reinforcement
Learning-based Safe Gradient Flow (RL-SGF).  At every iteration,
RL-SGF obtains estimates of the value functions and their respective
gradients associated with the objective function and safety constraint
of the constrained RL problem.  Next, RL-SGF updates the policy
parameters by solving a convex quadratically constrained quadratic
program that uses such estimates, and has a closed-form solution.  We
formally show that RL-SGF meets the desired specifications.  First, we
establish two results of independent interest that quantify the
statistical properties of our proposed estimates, including a bound on
their variance and the probability that the distance between the
estimates and the true values is within a tolerance.  Next, we
show that for any prescribed confidence, if the estimates of the
value functions and their gradients are generated with a sufficiently
large number of episodes (with an explicitly computable lower bound),
the next iterate of a safe policy under RL-SGF is a safe policy
with a probability higher than the prescribed confidence.  Moreover,
we show that RL-SGF asymptotically converges to a neighborhood of a
KKT point and that the size of this neighborhood can be made
arbitrarily small by using a sufficiently large number of episodes.
Finally, we illustrate the performance of RL-SGF in a navigation
example.

\section{Notation}

We denote by $\mathbb{Z}_{>0}$, $\real$, and $\real_{\geq0}$ the set
of positive integers, real, and nonnegative real numbers,
respectively.  Given $N\in\mathbb{Z}_{>0}$, we let
$[N]=\{1,2,\hdots,N \}$.  For $x\in\real^n$, $\norm{x}$ denotes its
Euclidean norm.  We let $\mathbf{I}_n$ be the $n$-dimensional identity matrix and
$\text{tanh}(x) = \frac{e^{x}-e^{-x}}{e^{x}+e^{-x}}$.  For a set
$B$ and $x\in\real^n$,
$\text{dist}(x,B) = \inf\limits_{y\in B} \norm{x-y}$.
Given random variables $X,Y$ taking scalar values, $\mathbb{E}[X]$
denotes the expectation of $X$,
$\text{Var}(X) = \mathbb{E}[(X-\mathbb{E}[X])^2]$ denotes its
variance.
%
%
Let $\Sc$ be a set of states, $\Ac$ a set of actions, and
$P:\Sc\times\Ac\times\Sc\to[0,1]$ a probability transition function,
where $P(s, a, s^{\prime})$ represents the probability that the agent
transitions to state $s^{\prime}\in\Sc$ given that it is at state
$s\in\Sc$ and takes action $a\in\Ac$.  We further let
$R_0:\Sc\times\Ac\times\Sc\to\real$ and
$R_1:\Sc\times\Ac\times\Sc\to\real$.  For every $s\in\Sc$, $a\in\Ac$,
and $s^{\prime}\in\Sc$, $R_0(s,a,s^{\prime})$
  %
  %
is the reward associated with completing a task when an agent is at
state $s$, takes action $a$, and transitions to state $s^{\prime}$.
Instead, $R_1(s,a,s^{\prime})$ is the cost associated with a safety
constraint when an agent is at state $s$, takes action $a$, and transitions to
state $s^{\prime}$.  We refer to the tuple $(\Sc,\Ac,P,R_0,R_1)$ as a
Constrained Markov Decision Process (CMDP).
%
%
A policy $\pi$ for the CMDP is a function that maps every state
$s\in\Sc$ to a distribution over the set of actions $\Ac$.  Such
distribution is denoted as $\pi(\cdot|s)$, and $\pi(a|s)$ is the
probability of taking action $a\in\Ac$ at state $s\in\Sc$.

\section{Problem Statement}\label{sec:problem-statement}

Given a CMDP $\mathcal{M}=(\mathcal{S},\mathcal{A},P,R_0,R_1)$, the
goal of the agent is to maximize the cumulative rewards while keeping
the cumulative costs below a certain threshold.  We consider a
parametric class of policies indexed by a vector
$\theta\in\real^d$. We denote the policy associated with $\theta$ as
$\pi_{\theta}$.  Given a distribution $\eta$ of initial states, a
discount factor $\gamma\in(0,1)$, and a time horizon
$T\in\mathbb{Z}_{>0}$, we consider the following problem:
\begin{subequations}
  \begin{align}
    &\min_{\theta\in\real^d} \quad V_0(\theta) = \mathbb{E}_{a \sim
      \pi_{\theta}(\cdot|s), s_0 \sim \eta} \biggl[ \sum_{k=0}^{T}
      -\gamma^k R_0(s_k, a_k, s_{k+1})
      \biggr]\label{eq:safe-RL-problem-parametrized-policies-objective}
    \\ 
    & \; \text{s.t.} \quad \quad V_1(\theta) = \mathbb{E}_{a \sim
      \pi_{\theta}(\cdot|s), s_0 \sim \eta} \biggl[ \sum_{k=0}^{T}
      \gamma^k R_1(s_k, a_k, s_{k+1}) \biggr] \leq
      0\label{eq:safe-RL-problem-parametrized-policies-constraint}. 
  \end{align}
  \label{eq:safe-RL-problem-parameterized-policies}
\end{subequations}
Problem~\eqref{eq:safe-RL-problem-parameterized-policies} seeks to
find the policy $\pi_{\theta}$ that maximizes the expected cumulative
reward given by $R_0$ (note that, for convenience, we have changed the
sign of $R_0$ to
turn~\eqref{eq:safe-RL-problem-parameterized-policies} into a
minimization problem) over $T$ time steps and keeps the expected
cumulative reward given by $R_1$ over $T$ time steps below zero.  The
discount factor $\gamma$ determines how much future rewards are valued
compared to immediate rewards.  In general, both $V_0$ and $V_1$ are
non-convex, which makes
solving~\eqref{eq:safe-RL-problem-parameterized-policies} NP-hard.
Our goal is to develop an algorithm that converges to a KKT point
of~\eqref{eq:safe-RL-problem-parameterized-policies} and is anytime,
meaning that at every iteration, the
constraint~\eqref{eq:safe-RL-problem-parametrized-policies-constraint}
is satisfied with a prescribed confidence.

\section{RL-SGF: Safe RL via Anytime Constrained
  Optimization}\label{sec:proposed-solution-challenges}

Here we present the algorithm to
solve~\eqref{eq:safe-RL-problem-parameterized-policies} in an anytime
fashion.
At each iteration $i\in\mathbb{Z}_{>0}$, we consider the update law
for parameter $\theta_i$ given by $\theta_{i+1}=p(\theta_i)$, where
$p:\real^d\to\real^d$ is defined as:
\begin{subequations}
  \begin{align}
    p(\theta) =
    &\argmin{y\in\real^d}{\nabla V_0(\theta)^\top(y-\theta)
      + \frac{1}{2h}\norm{y-\theta}^2 }
    \\ 
    &\text{s.t.} \ \alpha h V_1(\theta) \! + \! \nabla V_1(\theta)^\top
      (y-\theta) \! + \! \frac{1}{2h}\norm{y-\theta}^2 \leq 0. 
  \end{align}
  \label{eq:p-definition}
\end{subequations}
Note that $p$ is well defined over
$\Fc=\setdef{\theta\in\real^d}{\exists y\in\real^d \ \text{s.t.} \
  \alpha h V_1(\theta) \! + \! \nabla V_1(\theta)^\top (y-\theta) \!
  + \! \frac{1}{2h}\norm{y-\theta}^2 \leq 0}$.  The iterates defined
through~\eqref{eq:p-definition} are inspired by considering
discretizations of the Safe Gradient Flow (SGF)~\citep{AA-JC:24-tac},
which is a continuous-time algorithm for constrained optimization. We
note that~\eqref{eq:p-definition} is also a special case of the Moving
Balls Algorithm (MBA), introduced in~\citep{AA-RS-MT:10}.
Both SGF and MBA are anytime algorithms,
and~\citep{AA-JC:24-tac,AA-RS-MT:10} provide conditions under which
they converge to KKT points of the corresponding optimization problem.
We summarize the relevant properties of the update law $p$ next:
  
\begin{lemma}\longthmtitle{Constraint satisfaction, KKT points, and
      convergence}\label{lem:p-properties}
    Let $V_0$ and $V_1$ be Lipschitz on their domain of definition
    with Lipschitz constants $L_0$ and $L_1$ respectively, and assume
    $V_0$ is lower bounded.
    Let $\alpha>0$ and
    $h\in(0,\min\{ \frac{1}{\alpha},\frac{1}{L_0},\frac{1}{L_1} \} )$.
    For $\theta\in\real^d$ with $V_1(\theta)\leq0$, we have
  \begin{enumerate}
    \item\label{it:constraint-satisfaction} $V_1(p(\theta))\leq0$,
    \item\label{it:kkt-points} $\theta$ is a KKT point
      of~\eqref{eq:safe-RL-problem-parameterized-policies} if and only if
      $p(\theta)=\theta$,
    \item\label{it:convergence} any limit point of the sequence $\{ \theta_k \}_{k\in\mathbb{Z}_{>0}}$ defined as 
    $\theta_0=\theta$ and $\theta_k=p(\theta_{k-1})$ for $k\geq1$ is a KKT point 
    of~\eqref{eq:safe-RL-problem-parameterized-policies}.
  \end{enumerate}
\end{lemma}
\begin{proof}
  \ref{it:constraint-satisfaction}: Since $V_1$ is Lipschitz with Lipschitz constant $L_1$, 
  by~\citep[Lemma 1.2.3]{YN:18} we have 
  \begin{align}\label{eq:V1p-inequality}
      V_1(p(\theta)) \! \leq \! V_1(\theta) \! + \! \nabla
      V_1(\theta)^\top ( p(\theta)-\theta ) \! + \!
      \frac{L_1}{2}\norm{p(\theta)-\theta}^2.
  \end{align}
  Now, by definition, $p(\theta)$ satisfies 
  \begin{align*}
      \alpha h V_1(\theta) + \nabla V_1(\theta)^\top (p(\theta)-\theta) + \frac{1}{2h}\norm{ p(\theta)-\theta }^2 \leq 0,
  \end{align*}
  which together with~\eqref{eq:V1p-inequality} implies
  \begin{align*}
      V_1(p(\theta)) \leq (1-\alpha h)V_1(\theta) - \frac{1}{2}(\frac{1}{h}-L_1)\norm{ p(\theta)-\theta }^2 \leq 0,
  \end{align*}
  where in the last inequality we have used the fact that $h<\frac{1}{L_1}$, $h<\frac{1}{\alpha}$,
  and $V_1(\theta)\leq0$.

  \ref{it:kkt-points}: By definition, $p(\theta)$ solves the KKT optimality conditions for~\eqref{eq:p-definition},
  which read 
  as follows: there exists $u\in\real_{\geq0}$ such that 
  \begin{align*}
      &\nabla V_0(\theta) + u\nabla V_1(\theta) + \frac{1}{h}(1+u)(p(\theta)-\theta) = 0, \\
      &u(\alpha h V_1(\theta) + \nabla V_1(\theta)^T (p(\theta)-\theta) + \frac{1}{2h}\norm{p(\theta)-\theta}^2 ) = 0, \\
      &\alpha h V_1(\theta) + \nabla V_1(\theta)^T (p(\theta)-\theta) + \frac{1}{2h}\norm{p(\theta)-\theta}^2 \leq 0.
  \end{align*}
  Since $p(\theta)=\theta$, $\alpha>0$, and $h>0$, the above equations imply that $\theta$ satisfies the KKT 
  conditions for~\eqref{eq:safe-RL-problem-parameterized-policies}.

  \ref{it:convergence}: Let $k>1$. If $\theta_k = p(\theta_k)$, then $\theta_l=\theta_k$ for all $l\geq k$ and 
  $\{ \theta_k \}_{k\in\mathbb{Z}_{>0}}$ converges to a KKT point by~\ref{it:kkt-points}.
  If $\theta_k \neq p(\theta_k)$, since $V_0$ is Lipschitz with Lipschitz constant $L_0$, 
  by~\citep[Lemma 1.2.3]{YN:18} we have 
  \begin{align}\label{eq:V0p-inequality}
      V_0(p(\theta_k)) \! \leq \! V_0(\theta_k) \! + \! \nabla
      V_0(\theta_k)^\top ( p(\theta_k)-\theta_k ) \! + \!
      \frac{L_0}{2}\norm{p(\theta_k)-\theta_k}^2.
  \end{align}
  By the optimality conditions of~\eqref{eq:hat-p-definition}, it follows
  that $(y-p(\theta_k))^\top \widehat{\nabla V_0}(\theta_k) +
  \frac{1}{h}(p(\theta_k)-\theta_k) \geq 0$ for all
  $y\in\real^d$ that are feasible for~\eqref{eq:p-definition}~\citep[Theorem 12.3]{JN-SW:06}.
  Since $\theta_k$ is feasible for~\eqref{eq:p-definition} (when $\theta=\theta_k$), 
  we have 
  $(\theta_k-p(\theta_k))^\top \widehat{\nabla V_0}(\theta_k) +
  \frac{1}{h}(p(\theta_k)-\theta_k) \geq 0$.
  After substituting in~\eqref{eq:V0p-inequality}, we have 
  \begin{align}\label{eq:V0p-inequality-2}
      V_0(p(\theta_k)) \! \leq \! V_0(\theta_k)
      -\frac{1}{2}( \frac{1}{h} - L_0 )\norm{p(\theta_k)-\theta_k}^2.
  \end{align}
  Now, since $V_1(\theta)\leq0$, by~\ref{it:constraint-satisfaction} we have 
  $V_1(\theta_k)\leq0$ for all $k \geq 1$.
  Therefore, by~\eqref{eq:V0p-inequality-2} we have 
  \begin{align}\label{eq:V0p-inequality-3}
      V_0(p(\theta_k)) \! \leq \! -\frac{1}{2}( \frac{1}{h} - L_0 )\norm{p(\theta_k)-\theta_k}^2,
  \end{align}
  and hence $\{ V_0(\theta_k) \}_{k\in\mathbb{Z}_{>1}}$ is monotonically decreasing.
  Since $V_0$ is bounded, $\{ V_0(\theta_k) \}_{k\in\mathbb{Z}_{>1}}$ is necessarily convergent.
  Let $V_0^* = \lim\limits_{k\to\infty}  V_0(\theta_k)$.
  By taking limits on both sides of~\eqref{eq:V0p-inequality-2}, and since $V_0$ and 
  $p$ are continuous (cf.~\citep[Proposition 2.3 (ii)]{AA-RS-MT:10}), 
  it follows that each limit point $x_{\infty}$ of the sequence 
  $\{ V_0(\theta_k) \}_{k\in\mathbb{Z}_{>1}}$ satisfies 
  $x_{\infty} = p(x_{\infty})$ and by~\ref{it:kkt-points} 
  is therefore a KKT point of~\eqref{eq:safe-RL-problem-parameterized-policies}.
\end{proof}

We note that
Lemma~\ref{lem:p-properties}~\ref{it:constraint-satisfaction} ensures
that the update rule $\theta_{i+1}=p(\theta_i)$ is anytime.
Note also that
$\setdef{\theta\in\real^d}{V_1(\theta)\leq0}\subset\Fc$, which ensures
that $p$ is well-defined for all safe policies.
%
%
Although Lemma~\ref{lem:p-properties} ensures that the update rule $\theta_{i+1}=p(\theta_i)$ 
is anytime and leads to convergence to KKT points of~\eqref{eq:safe-RL-problem-parameterized-policies},
the main difficulty in computing $p(\theta)$
in~\eqref{eq:p-definition} is that $V_1(\theta)$,
$\nabla V_0(\theta)$, and $\nabla V_1(\theta)$ are unknown, and need
to be estimated.
%
%
%
Before introducing such estimates, we make the following assumptions.
\begin{assumption}\longthmtitle{Boundedness of reward
    functions}\label{as:boundedness-reward}
  There exist $B_0,B_1>0$ such that $|R_0(s,a,s^{\prime})|<B_0$ and
  $|R_1(s,a,s^{\prime})|<B_1$, for all $s\in\mathcal{S}$,
  $a\in\mathcal{A}$, and $s^{\prime}\in\mathcal{S}$.
\end{assumption}

\begin{assumption}\longthmtitle{Differentiability and Lipschitzness of
    policy}\label{as:differentiability-lipschitzness-policy}
  For any $a\in\Ac$, $s\in\Sc$, and $\theta\in\real^d$, 
  $\pi_{\theta}(a|s)>0$.
  %
  %
  The function $\Phi_{a,s}:\real^d\to\real$ defined as 
  $\Phi_{a,s}(\theta) = \log \pi_{\theta}(a|s)$ is continuously differentiable
  and there exist positive constants $L$ and $\tilde{B}$ such that 
  \begin{align*}
    &\norm{\nabla \Phi_{a,s}(\theta_1)-\nabla \Phi_{a,s}(\theta_2)}
      \leq L\norm{\theta_1-\theta_2}, \ \forall \
      \theta_1,\theta_2\in\real^d,a\in\Ac,s\in\Sc,
    \\
    &\norm{ \frac{ \partial \Phi_{a,s}(\theta) }{ \partial \theta_i }
      } \leq \tilde{B}, \ \forall\theta\in\real^d, i\in[d],
      a\in\Ac,s\in\Sc. 
  \end{align*}
\end{assumption}

Assumptions~\ref{as:boundedness-reward}
and~\ref{as:differentiability-lipschitzness-policy} are standard in
the literature (cf.~\citep{KZ-AK-HZ-TB:20,QB-WUM-VA:24}).  By the
Policy Gradient Theorem~\citep[Section 13.2]{RSS-AGB:18}, under
Assumption~\ref{as:differentiability-lipschitzness-policy}, the
functions $V_0$ and $V_1$
in~\eqref{eq:safe-RL-problem-parameterized-policies} are
differentiable.  Moreover, Lemma~\ref{lem:lipschitzness-value-functions}
ensures that $V_0$ and $V_1$ are globally Lipschitz,
%
%
with Lipschitz constants that can be computed in terms of $L$, $B_0$,
$B_1$, $\tilde{B}$, $\gamma$, and $T$.  In the sequel, we let
$L_0$ and $L_1$ be the Lipschitz constants of $V_0$ and $V_1$,
respectively.

We are now ready to introduce unbiased estimates of $V_1(\theta)$,
$\nabla V_0(\theta)$, and $\nabla V_1(\theta)$, based on the Policy
Gradient Theorem~\citep[Section 13.2]{RSS-AGB:18}. We also provide
upper bounds on the variances of such estimates, and probabilistic
upper bounds on the distance to their respective means.  Throughout
the paper, all probabilities, expectations and variances are taken
with respect to $a\sim\pi_{\theta}(\cdot|s)$ and $s_0\sim\eta$, and
from now on we do not denote this explicitly for the sake of
simplicity.

\begin{lemma}\longthmtitle{Statistical properties of estimates of the
    value function and their
    gradients}\label{lem:estimator-statistical-properties}
  Suppose that Assumptions~\ref{as:boundedness-reward}
  and~\ref{as:differentiability-lipschitzness-policy} hold.  Let
  $\theta\in\real^d$, $N\in\mathbb{Z}_{>0}$, and consider $N$
  independent episodes of the form
  $[s_0^n, a_0^n, s_1^n, a_1^n, \hdots, s_T^n, a_T^n, s_{T+1}^n]$
  generated with policy $\pi_{\theta}$, where $s_0^n\sim\eta$ and
  $a_i^n \sim \pi_{\theta}(\cdot | s_i^n)$ for all $i\in[T]$ and
  $n\in[N]$.  Let $q\in\{0,1\}$ and define the value function
  estimator
    %
  %
  $\widehat{V_q}(\theta) = \frac{(-1)^{q+1}}{N}\sum_{n=1}^N \sum_{t=0}^T
  \gamma^t R_q(s_t^n,a_t^n,s_{t+1}^n)$.
    %
  %
  Given a function $b:\mathcal{S}\to\real$ (referred to as
  \textit{baseline}) such that $|b(s)|\leq\hat{B}$ for all
  $s\in\mathcal{S}$, consider the value function gradient estimator
  $\widehat{\nabla V_q}(\theta) = \frac{1}{N}\sum_{n=1}^N\sum_{t=0}^T
  \gamma^t \nabla \Phi_{a_t^n,s_t^n}(\theta) \!
  \sum_{t^{\prime}=t}^{T} \! \Big( \gamma^{t^{\prime}-t}
  R_q(s_{t^\prime}^n,a_{t^\prime}^n,s_{t^{\prime}+1}^n) \! - \! b(s_t)
  \Big)$.
  Further define $\tilde{\sigma}_q = B_q \frac{1-\gamma^{T+1}}{1-\gamma}$
  and $\bar{\sigma}_q = \tilde{B} \sum_{t=0}^T \gamma^t \sum_{t^\prime=t}^T ( B_q \gamma^{t^\prime-t} + \hat{B} )$
  for $q\in\{0,1\}$. Then,
  \begin{enumerate}
  \item\label{it:value-function-estimator} The estimator
    $\widehat{V_q}(\theta)$ is unbiased, i.e.,
    $\mathbb{E}[\widehat{V_q}(\theta)] = V_q(\theta)$.  Moreover,
    $\Var(\widehat{V_q}(\theta)) \leq \frac{\tilde{\sigma}_q^2}{N}$
    and for all $\epsilon>0$,
    $\mathbb{P}(|\widehat{V_q}(\theta)-V_q(\theta)|\leq \epsilon) \geq
    1-\exp \Big\{ -\frac{N\epsilon^2}{2 \tilde{\sigma}_q^2 }
    \Big\}$.  \vspace{-0.4cm}
  \item\label{it:gradient-estimator} The estimator
    $\widehat{\nabla V_q}(\theta)$ is unbiased, i.e.,
    $\mathbb{E}[\widehat{\nabla V_q}(\theta)] = \nabla V_q(\theta)$.
    Moreover,
    $\Var(\widehat{\nabla V_q}(\theta)_i) \leq \frac{ \bar{\sigma}_q^2
    }{N}$ for all $i\in[d]$,
    and for all $\epsilon>0$,
    $\mathbb{P}\Big(\norm{ \nabla V_q(\theta) - \widehat{\nabla
        V_q}(\theta) } \leq \epsilon \Big) \geq 1-d\exp\Big\{
    -\frac{N\epsilon^2}{2d\bar{\sigma}_q^2} \Big\}$.
  \end{enumerate}
\end{lemma}
\begin{proof}
  \ref{it:value-function-estimator}: The estimator
  $\widehat{V_q}(\theta)$ is unbiased because
  $\mathbb{E} \Bigl[ \sum_{t=0}^T \gamma^t R_q(s_t^n,a_t^n,s_{t+1}^n)
  \Bigr] = V_q(\theta) $.  To compute the variance, note that for any
  $n\in[N]$, because of Assumption~\ref{as:boundedness-reward},
  $\bigg\rvert \sum_{t=0}^T \gamma^t R_q(s_t^n,a_t^n,s_{t+1}^n)
  \bigg\rvert \! \leq \! \sum_{t=0}^T \gamma^t
  |R_q(s_t^n,a_t^n,s_{t+1}^n)| \leq B_q
  \frac{1-\gamma^{T+1}}{1-\gamma} = \tilde{\sigma}_q$.
  By Popoviciu's inequality~\citep[Corollary 1]{RB-CD:00}, for any
  $n\in[N]$ we have
  $\Var \Big( \sum_{t=0}^T \gamma^t R_q(s_t^n,a_t^n,s_{t+1}^n) \Big)
  \leq \tilde{\sigma}_q^2$.  Since the $N$ episodes are independent,
  the bound on the variance of $\widehat{V_q}(\theta)$ follows.
  On the other hand, the lower bound on $\mathbb{P}(|\widehat{V_q}(\theta)-V_q(\theta)|\leq \epsilon)$
  follows from Hoeffding's inequality~\citep{WH:63-hoeffding}.
  \\
  ~\ref{it:gradient-estimator}: The fact that
  $\widehat{\nabla V_q}(\theta)$ is unbiased holds by the Policy
  Gradient Theorem with baseline (cf.~\citep[Section
  13.4]{RSS-AGB:18}).
  Now, note that for any $n\in[N]$ and $i\in[d]$,
  \begin{align*}
    &\Bigg\vert \sum_{t=0}^T  \gamma^t \frac{\partial \Phi_{a_t^n,
      s_t^n}(\theta) }{\partial\theta_i} \! \sum_{t^{\prime}=t}^{T}
      \! \Big( \gamma^{t^{\prime}-t} R_q(
      s_{t^{\prime}},a_{t^{\prime}},s_{t^{\prime}+1} ) \! - \!
      b(s_t) \Big) \Bigg\vert  
      \leq \tilde{B} \sum_{t=0}^T \gamma^t \sum_{t^{\prime}=t}^T (
        B_q \gamma^{t^{\prime}-t} + \hat{B} ) = \bar{\sigma}_q. 
  \end{align*}
  By Popoviciu's inequality~\citep[Corollary 1]{RB-CD:00}, and since
  the $N$ episodes are independent, it follows that
  $\Var( \widehat{\nabla V_q}(\theta)_i ) \leq
  \frac{\bar{\sigma}_q^2}{N}$ for all $i\in[d]$.
  Next, 
  by Hoeffding's inequality~\citep{WH:63-hoeffding},
  for any $\epsilon>0$ and $i\in[d]$,
  $\mathbb{P}( |\widehat{\nabla V_q}(\theta)_i - \nabla V_q(\theta)_i|
  \leq \epsilon ) \geq 1- \exp \Bigl\{
  -\frac{N\epsilon^2}{2\bar{\sigma}_q^2} \Bigr\}$.  By Fréchet's
  Inequality~\citep{MF:35},
  $\mathbb{P}\Big( \bigcap_{i=1}^d \Big\{ |\nabla V_q(\theta)_i \! -
  \! \widehat{\nabla V_q}(\theta)_i| \leq \frac{\epsilon}{\sqrt{d}}
  \Big\} \Big) \! \geq \! 1-d \exp \Bigl\{
  -\frac{N\epsilon^2}{2d\bar{\sigma}_q^2} \Bigr\}$.  Now, note that if
  $|\nabla V_q(\theta)_i - \widehat{\nabla V_q}(\theta)_i| \leq
  \frac{\epsilon}{\sqrt{d}}$ for all $i\in[d]$, we have that
  $\norm{\nabla V_q(\theta)-\widehat{\nabla V_q}(\theta)} \leq
  \epsilon$, which means that \vspace{-0.4cm}
  \begin{align*}
    \mathbb{P}\Big( \norm{ \nabla V_q(\theta) - \widehat{\nabla
    V_q}(\theta) } \leq \epsilon \Big) \geq  
    \mathbb{P}\Big( \bigcap_{i=1}^d \Big\{ |\nabla V_q(\theta)_i \! -
    \! \widehat{\nabla V_q}(\theta)_i| \leq \frac{\epsilon}{\sqrt{d}}
    \Big\}  \Big) \! 
    \geq 1-d\exp\Big\{ -\frac{N\epsilon^2}{2d\bar{\sigma}_q^2} \Big\}. 
  \end{align*}
  \vspace{-0.4cm}
\end{proof}

A bound on the variance similar to that of
Lemma~\ref{lem:estimator-statistical-properties}~\ref{it:gradient-estimator}
appears in~\citep{TZ-HH-GN-MS:12} for Gaussian policies.
Given the estimators introduced in
Lemma~\ref{lem:estimator-statistical-properties},
%
%
we propose the update law $\theta_{i+1} = \hat{p}(\theta_i)$, where
$\hat{p}:\real^d\to\real^d$ is obtained by replacing
$V_1, \nabla V_1$, and $\nabla V_0$ in~\eqref{eq:p-definition} by
their estimates, i.e.,
\begin{subequations}
  \begin{align}
    \hat{p}(\theta) =
    &\argmin{y\in\real^d}{ \widehat{\nabla
      V_0}(\theta)^\top(y-\theta) +
      \frac{1}{2h}\norm{y-\theta}^2 }
    \\ 
    &\text{s.t.} \ \alpha h\widehat{V_1}(\theta) \! + \!
      \widehat{\nabla V_1}(\theta)^\top \! (y-\theta) \! + \!
      \frac{1}{2h}\norm{y-\theta}^2 \leq 0. 
  \end{align}
  \label{eq:hat-p-definition}
\end{subequations}
Note that $\hat{p}$ is well defined on
$\hat{\Fc}:=\setdef{\theta\in\real^d}{\exists y\in\real^d \
  \text{s.t.} \ \alpha h \widehat{V_1}(\theta) \! + \!
  \widehat{\nabla V_1}(\theta)^\top (y-\theta) \! + \!
  \frac{1}{2h}\norm{y-\theta}^2 \leq 0}$. Furthermore,
$\setdef{\theta\in\real^d}{\widehat{V_1}(\theta) \leq 0}\subset
\hat{\Fc}$, which ensures that $\hat{p}$ is defined for all policies
that are estimated to
satisfy~\eqref{eq:safe-RL-problem-parametrized-policies-constraint}.
Our approach is summarized in Algorithm~\ref{alg:anytime-safe-rl}, and
since it takes inspiration from considering discretizations of the
Safe Gradient Flow (SGF)~\citep{AA-JC:24-tac}, we refer to it as
Reinforcement Learning-based Safe Gradient Flow (RL-SGF).

\begin{algorithm}
  \caption{ RL-SGF }\label{alg:anytime-safe-rl}
  \begin{algorithmic}
    \STATE \textbf{Parameters}: $\alpha$, $h$, $T$, $k$, and a sequence $\{ N_i \}_{i=1}^k$,
    \textbf{Initial Policy Parameter}: $\theta_1$;
    \FOR{$i\in[1,\hdots,k]$}
    \STATE Generate $N_i$ episodes of length $T+1$ using policy $\pi_{\theta_i}$;
    \STATE Compute $\widehat{V_1}(\theta_i)$ as defined Lemma~\ref{lem:estimator-statistical-properties}~\ref{it:value-function-estimator}
    and $\widehat{\nabla V_1}(\theta_i)$, $\widehat{\nabla V_0}(\theta_i)$ as defined in Lemma~\ref{lem:estimator-statistical-properties}~\ref{it:gradient-estimator};
    \STATE $\theta_{i+1} \leftarrow \hat{p}(\theta_i)$;
    \ENDFOR
  \end{algorithmic}
\end{algorithm}
%
%

\vspace{-0.25cm}

\section{Anytime Safety and Convergence Properties of RL-SGF}

In this section we study the anytime and convergence properties of
Algorithm~\ref{alg:anytime-safe-rl}.  We fix $\alpha>0$,
$h \in(0,\min\{ \frac{1}{L_0}, \frac{1}{L_1}, \frac{1}{\alpha} \} )$,
and assume that for any $\theta\in\real^d$, the estimates
$\widehat{V_1}(\theta)$, $\widehat{\nabla V_0}(\theta)$, and
$\widehat{\nabla V_1}(\theta)$ are computed according to
Lemma~\ref{lem:estimator-statistical-properties} with $N$ different
episodes of length $T+1$ of the policy $\pi_{\theta}$.
First we provide a closed-form expression for the update maps
$\hat{p}$ and~$p$.

\smallskip
\begin{lemma}\longthmtitle{Closed-form of update 
    maps}\label{lem:closed-form-expression-next-iterate}
  Let $\theta\in\real^d$ and suppose that Slater's condition holds
  for~\eqref{eq:hat-p-definition}~\footnote{ I.e.,
  $\exists \ y\in\real^d$ with $\alpha h \widehat{V_1}(\theta) + 
  \widehat{\nabla V_1}(\theta)^\top (y-\theta)+\frac{1}{2h}\norm{y-\theta}^2<0$
  (cf.~\cite[Section 5.2.3]{SB-LV:09}) }.
  %
  %
  Define
  $\hat{A}_{\theta} = \norm{\widehat{\nabla V_1}(\theta)}^2
  -2\alpha\widehat{V_1}(\theta)$,
  $\hat{B}_{\theta} = 2 \hat{A}_{\theta}$,
  $\hat{C}_{\theta} = 2 \widehat{\nabla V_1}(\theta)^\top
  \widehat{\nabla V_0}(\theta) \! - \! \norm{\widehat{\nabla
      V_0}(\theta)}^2 \! - \! 2\alpha\widehat{V_1}(\theta)$,
  $\hat{\Delta}_{\theta} = 4 \norm{\widehat{\nabla V_1}(\theta) \! -
    \! \widehat{\nabla V_0}(\theta)}^2 \hat{A}_{\theta}$, and
  $\hat{u}_{\theta} =
  \frac{-\hat{B}_{\theta}+\sqrt{\hat{\Delta}_{\theta}}}{2\hat{A}_{\theta}}$
  if $\hat{C}_{\theta}<0$, and $\hat{u}_{\theta} =0$ otherwise.
  Then, if $\hat{A}_{\theta} \! > \! 0$,
  $\hat{p}(\theta) \! = \! \theta \! - \! h\frac{ \widehat{\nabla V_0}(\theta) + \hat{u}_{\theta}\widehat{\nabla V_1}(\theta) }{ 1 + \hat{u}_{\theta} }$,
  and if $\hat{A}_{\theta} \! = \! 0$, 
  $\hat{p}(\theta) \! = \! \theta \! - \! h\widehat{\nabla V_1}(\theta)$.
    If Slater's condition holds for~\eqref{eq:p-definition},
    an analogous expression holds for $p$, with the estimates
    replaced by their true values.
\end{lemma}
  \begin{proof}
      Since~\eqref{eq:hat-p-definition} is strongly convex and Slater's condition is satisfied,
      strong duality holds for~\eqref{eq:hat-p-definition} for 
      any $\theta$ such that $\widehat{V_1(\theta)}\leq0$.
      Let $\Lc:\real^d\times\real_{\geq0}\times\real^d$ be the Lagrangian 
      of~\eqref{eq:hat-p-definition}, i.e.,
      \begin{align*}
          \Lc(y,u,\theta) &= \widehat{\nabla V_0(\theta)}^T(y-\theta) + \frac{1}{2h}\norm{y-\theta}^2 + u\Big( \alpha h \widehat{V_1(\theta)} + \widehat{\nabla V_1(\theta)}^T (y-\theta) + \frac{1}{2h}\norm{y-\theta}^2 \Big).
      \end{align*}
      Let $y^*:\real\times\real^d\to\real^d$ be the function mapping every $(u,\theta)\in\real_{\geq0}\times\real^d$ 
      to the minimizer of $\Lc$ in $y$ for that fixed $(u,\theta)$. 
      Since $\Lc$ is quadratic in $y$, the minimizer can be computed in closed-form and is given by
      $y^*(u,\theta) = \theta  - h\frac{ \widehat{\nabla V_0(\theta)} + u\widehat{\nabla V_1(\theta)} }{ 1 + u }$.
      Now, $\hat{p}(\theta)$ is given by $\hat{p}(\theta) = y^*(\hat{u},\theta)$,
      where $\hat{u}$ 
      is the minimizer of the dual problem (cf.~\citep[Section 5.2.3]{SB-LV:09}), i.e., 
      $\hat{u} = \argmax{u\in\real_{\geq0}}{\Lc(y^*(u,\theta),u,\theta)}$,
      which is equivalent to
      $\hat{u} = \argmin{u\in\real_{\geq0}}{\ell(u)}$,
      where 
      $\ell(u) = h\frac{ \norm{ \widehat{\nabla V_0(\theta) } + 
      u\widehat{\nabla V_1(\theta)} }^2 }{2(1+u)} - u \alpha h \widehat{V_1(\theta)}$.
      Now, note that the critical points of $\ell$ can be obtained by finding the roots of the quadratic equation
      $\hat{A}_{\theta} u^2 + \hat{B}_{\theta} u + \hat{C}_{\theta} = 0$,
      which has discriminant $\hat{B}_{\theta}^2-4\hat{A}_{\theta}\hat{C}_{\theta}$ equal to 
      $\hat{\Delta}_{\theta}$.
      Note that if $\hat{A}_{\theta}>0$, the quadratic coefficient is positive, 
      and also $\hat{B}_{\theta}>0$ and $\hat{\Delta}_{\theta}\geq0$.
      If $\hat{C}_{\theta}\geq0$, $\sqrt{\hat{\Delta}_{\theta}} \leq |\hat{B}_{\theta}|$, which means that 
      $\frac{-\hat{B}_{\theta}+\sqrt{\hat{\Delta}_{\theta}} }{2\hat{A}_{\theta}} < 0$. This implies that 
      $\hat{A}_{\theta}^2 u^2 + \hat{B}_{\theta}u + \hat{C}_{\theta} > 0$ for all $u\geq0$, which in turn means 
      that the function $\ell$
      is increasing for all $u\geq0$ and hence $\hat{u}=0$.
      On the other hand, if $\hat{A}_{\theta}=0$, then $\hat{B}_{\theta}=0$, and 
      $\hat{C}_{\theta}=2 \widehat{\nabla V_1}(\theta)^\top
      \widehat{\nabla V_0}(\theta) \! - \! \norm{\widehat{\nabla
      V_0}(\theta)}^2 \! - \norm{\widehat{\nabla V_1(\theta)}}^2=-\norm{ \widehat{\nabla V_1}(\theta) - \widehat{\nabla V_0}(\theta) }^2\leq0$.
      If $\hat{C}_{\theta}=0$, $\widehat{\nabla V_0(\theta)}=\widehat{\nabla V_1}(\theta)$, and 
      $p(\theta)=\theta-\widehat{\nabla V_1(\theta)}$.
      Instead, if $C_{\theta}<0$, then $\ell^\prime(u)<0$ for all $u\geq0$ and 
      $\hat{u}_{\theta}=\infty$, which also means that 
      $p(\theta)=\theta-\widehat{\nabla V_1}(\theta)$.
      By following an analogous argument using the true values of the value functions and their gradients, 
      instead of their estimates, we obtain an analogous expression for $p$. 
    \end{proof}


Note that if $\widehat{V_1}(\theta) \leq -\eta$ 
(for some $\eta\geq0$), then $\hat{A}_{\theta} \geq 2\alpha\eta\geq0$.
In the sequel, we denote by $A_{\theta}$, $B_{\theta}$, $C_{\theta}$,
$\Delta_{\theta}$, and $u_{\theta}$ the quantities analogous to
$\hat{A}_{\theta}$, $\hat{B}_{\theta}$, $\hat{C}_{\theta}$,
$\hat{\Delta}_{\theta}$, and $\hat{u}_{\theta}$, as defined in
Lemma~\ref{lem:closed-form-expression-next-iterate} but substituting
the estimates
by their true values.


\begin{remark}\longthmtitle{Connection with primal-dual methods}{ \rm
  Primal-dual methods, see e.g.,~\cite{SP-LC-MCF-AR:19}, employ the
  Lagrangian function associated to the constrained optimization
  problem~\eqref{eq:safe-RL-problem-parameterized-policies},
  updating the primal variable $\theta$ via gradient descent and the
  dual variable $\lambda$ through gradient ascent.  If the value of
  the dual variable at iteration $i\in\mathbb{Z}_{>0}$ is
  $\lambda_i$, the primal update is
  $\theta_{i+1} = \theta_i - \eta ( \widehat{\nabla V_0}(\theta) +
  \lambda_i \widehat{\nabla V_1}(\theta))$ (here $\eta$ is a
  predefined stepsize).  One can see the parallelism, cf.
  Lemma~\ref{lem:closed-form-expression-next-iterate}, with the
  policy update performed by RL-SGF,
  $ \hat p(\theta) = \theta - \frac{ h }{ 1 + \hat{u}_{\theta} } (
  \widehat{\nabla V_0}(\theta) + \hat{u}_{\theta} \widehat{\nabla
    V_1}(\theta)) $ (note how the other two expressions are
  recovered from this one by setting $\hat{u}_{\theta}= 0$ and
  $\hat{u}_{\theta} = \infty$, respectively).  When comparing both
  algorithms, we can interpret the RL-SGF update rule as a
  modification of the primal step of the primal-dual algorithm,
  where the stepsize $\eta$ is state-dependent and takes the value
  $\eta = \frac{h}{1 + \hat{u}_\theta}$, and
  $\lambda_i = \hat{u}_\theta$. This is consistent with the fact
  that the Safe Gradient Flow~\cite[Section IV.A]{AA-JC:24-tac}
  admits a primal-dual interpretation, where control inputs play the
  role of approximators of the dual variables. The key difference is
  that, while primal-dual methods use a fixed stepsize, RL-SGF
  dynamically adjusts $\eta_\theta$ and $\lambda_i$ at each
  iteration in a state-dependent fashion, in order to ensure that
  the next update remains feasible. This eliminates the need for
  fine-tuning hyperparameters. Our experiments suggest that this
  adaptive approach improves algorithm performance, see
  Figure~\ref{fig:PD-RF-SGF-comparison} below.  Additionally, the
  initialization of the dual variable is critical for the overall
  safety performance of the algorithm. For example, if the dual
  variable is initialized below the optimal value, the primal-dual
  algorithm would inevitably violate the safety constraint during
  the training process, whereas RL-SGF does not require initializing
  any dual variable and is guaranteed to be safe as long as the
  initial policy is safe.  \demo}
\end{remark}

The following result utilizes
Lemma~\ref{lem:estimator-statistical-properties} to provide safety
guarantees for Algorithm~\ref{alg:anytime-safe-rl}.

\begin{proposition}\longthmtitle{Safety guarantees of RL-SGF}\label{prop:safety-guarantees}
  Let $k\in\mathbb{Z}_{>0}$, $i\in[k]$, $\delta\in(0,1)$, and
  $\bar{\sigma}_1$, $\tilde{\sigma}_1$ be as in
  Lemma~\ref{lem:estimator-statistical-properties}. Suppose that
  Assumptions~\ref{as:boundedness-reward}
  and~\ref{as:differentiability-lipschitzness-policy} hold.
  For the $i$th iterate   $\theta_i \in \real^d$ of
  Algorithm~\ref{alg:anytime-safe-rl},
  \begin{enumerate}[leftmargin=0.6cm]
  \item\label{it:safety-guarantees-first} If
    $\widehat{V_1}(\theta_i) \! \leq \! 0$, let
    $N_i \! > \! \max\Bigl\{ \frac{-2 \tilde{\sigma}_1^2 \! \log(\delta) }{
      \hat{M}_i^2 },
      \! \frac{-2 d \bar{\sigma}_1^2
      \! \log(\frac{\delta}{d}) }{ \hat{M}_i^2 } \Bigr\}$, with
    $\hat{M}_i \! = \!  \frac{(1-\alpha h)|\widehat{V_1}(\theta_i)| +
      \frac{1}{2}(\frac{1}{h}-L_1)\norm{ \hat{p}(\theta_i)-\theta_i
      }^2}{1 + \norm{\hat{p}(\theta_i)-\theta_i}}$.

  \item\label{it:safety-guarantees-second} If
    $\widehat{V_1}(\theta_i) \! \geq \! 0$, let
    $N_i \! > \! \max\Bigl\{ \frac{-2 \tilde{\sigma}_1^2 \log(\delta) }{
      \nu^2 }, \frac{-2 d \bar{\sigma}_1^2 \log(\frac{\delta}{d}) }{
      \nu^2 } \Bigr\}$, with $\nu>0$ such that
    \vspace{-0.2cm}
    \begin{align}\label{eq:nu-condition}
      \nu(1+\norm{\hat{p}(\theta_i)-\theta_i}) + (1-\alpha
      h)\widehat{V_1}(\theta_i) < \frac{ ( \frac{1}{h}-L_1 )
      \norm{\hat{p}(\theta_i)-\theta_i}^2 }{2}. 
    \end{align} 
  \end{enumerate}
  In either case, we have
  $\mathbb{P}(V_1(\hat{p}(\theta_i)\leq0)) \geq 1-2\delta$.
  \smallskip
\end{proposition}
\begin{proof}
  \ref{it:safety-guarantees-first}: Since $\nabla V_1$ has Lipschitz constant $L_1$, and dropping 
  the dependency on $\theta_i$ for simplicity, by~\citep[Lemma 1.2.3]{YN:18}, and using 
  the Cauchy-Schwartz inequality, we have
  \vspace{-0.1cm}
  \begin{align}\label{eq:v-bound}
    &V_1(\hat{p}) \leq
      \widehat{V_1} \! + \! V_1 \! - \! \widehat{V_1} \! + \! \widehat{\nabla V_1}^\top ( \hat{p}-\theta_i )
      \! + \! \norm{ \nabla V_1 \! - \! \widehat{\nabla V_1} } \norm{ \hat{p} \! - \! \theta_i } \! + \! \frac{L_1}{2}\norm{\hat{p}\!-\!\theta_i}^2.
  \end{align}
  \vspace{-0.1cm}
  From~\eqref{eq:p-definition} we have,
  $\alpha h \widehat{V_1}(\theta_i) + \widehat{\nabla
    V_1}(\theta_i)^\top ( \hat{p}(\theta_i)-\theta_i ) \leq
  -\frac{1}{2h}\norm{\hat{p}(\theta_i)-\theta_i}^2$, which
  from~\eqref{eq:v-bound} implies
  \begin{align}\label{eq:v-bound-2}
    &V_1(\hat{p}) \leq 
      (1-\alpha h)\widehat{V_1} \! + \! V_1 \! - \! \widehat{V_1}
      \! + \! \norm{ \nabla V_1 \! - \! \widehat{\nabla V_1} } \norm{
      \hat{p} \! - \! \theta_i } \! - \!
      \frac{1}{2}(\frac{1}{h}-L_1)\norm{\hat{p}\!-\!\theta_i}^2. 
  \end{align}
  \vspace{-0.15cm}
  Now, by
  Lemma~\ref{lem:estimator-statistical-properties}~\ref{it:value-function-estimator},
  if $N_i > \frac{-2 \tilde{\sigma}_1^2 \log(\delta) }{\hat{M}_i^2}$,
  $\mathbb{P}(|\widehat{V_1}(\theta_i)-V_1(\theta_i)| \leq \hat{M}_i) \geq 1-\delta$.
  Similarly, by Lemma~\ref{lem:estimator-statistical-properties}~\ref{it:gradient-estimator},
  if
  $N_i \! > \!  \frac{-2 d \bar{\sigma}_1^2 \log(\frac{\delta}{d}) }{
    \hat{M}_i^2 }$,
  $\mathbb{P}\Big(\norm{ \widehat{\nabla V_1}(\theta_i) - \nabla
    V_1(\theta_i) } \leq \hat{M}_i \Big) \! \geq \! 1-\delta$.
  By~\eqref{eq:v-bound-2}, if
  $|\widehat{V_1}(\theta_i)\!-\!V_1(\theta_i)| \! \leq \! \hat{M}_i$,
  $\norm{\widehat{\nabla V_1}(\theta_i)\!-\!\nabla V_1(\theta_i)} \! \leq \!
  \hat{M}_i$, and $\widehat{V_1}(\theta_i) \! \leq \! 0$, we have
  $V_1(\hat{p}(\theta_i)) \! \leq \! 0$. Then, since $N_i$ satisfies 
  the bound in~\ref{it:safety-guarantees-first}, the result follows by   
  Fréchet's Inequality~\citep{MF:35}. \\
  \ref{it:safety-guarantees-second}: If $\nu>0$ satisfies~\eqref{eq:nu-condition},
  $|\widehat{V_1}(\theta_i)-V_1(\theta_i)| \leq \nu$, and
  $\norm{\widehat{\nabla V_1}(\theta_i)-\nabla V_1(\theta_i)} \leq
  \nu$, then by~\eqref{eq:v-bound}, $V_1(\hat{p}(\theta_i)) \leq 0$
  even if $\widehat{V_1}(\theta_i) > 0$. By following a similar
  argument as in~\ref{it:safety-guarantees-first} we have that if
  $N_i$ satisfies the lower bound
  on~\ref{it:safety-guarantees-second},
  $\mathbb{P}( V_1(\hat{p}(\theta_i)\leq0) ) \geq 1-2\delta$.
\end{proof}

Proposition~\ref{prop:safety-guarantees}~\ref{it:safety-guarantees-first}
shows that if $\widehat{V_1}(\theta_i)\leq0$ (i.e., we estimate that
the $i$th iterate is safe), then by running a sufficiently large
number of episodes the next iterate of
Algorithm~\ref{alg:anytime-safe-rl} is safe with an arbitrarily high
probability.  On the other hand,
Proposition~\ref{prop:safety-guarantees}~\ref{it:safety-guarantees-second}
provides similar guarantees when we estimate that the policy is unsafe
(i.e., $\widehat{V_1}(\theta_i)\geq0$).  We note that if
$\widehat{ V_1}(\theta_i)\leq0$,~\eqref{eq:nu-condition}
is satisfied by taking $\nu=0$. By continuity, this suggests that~\eqref{eq:nu-condition} is also satisfied in a neighborhood of 
$\setdef{\theta\in\real^d}{\widehat{V_1}(\theta)\leq0}$, and it becomes increasingly 
difficult to satisfy as $\widehat{V_1}(\theta)$ grows, which indicates that 
safety can be ensured in the next iterate as long 
as the safety violation is not too large.

Note that the lower bounds on $N_i$ in
Proposition~\ref{prop:safety-guarantees} depend on
$\hat{V}_1(\theta_i)$, $\widehat{\nabla V_1}(\theta_i)$, and
$\widehat{\nabla V_0}(\theta_i)$, which in turn also depend on $N_i$.
Therefore, we use the result in
Proposition~\ref{prop:safety-guarantees} is as follows: given a number
of episodes $N_i$, we construct the estimates $\hat{V}_1(\theta_i)$,
$\widehat{\nabla V_1}(\theta_i)$, and
$\widehat{\nabla V_0}(\theta_i)$.  If the lower bounds in $N_i$ in
Proposition~\ref{prop:safety-guarantees} hold, the guarantees provided
therein apply.  If they do not, one should increase $N_i$. This
process is guarantees to terminate, because the estimates
$\widehat{V_1}(\theta_i)$, $\widehat{\nabla V_1}(\theta_i)$,
$\widehat{\nabla V_0}(\theta_i)$ converge to the true values
$V_1(\theta_i)$, $\nabla V_1(\theta_i)$, $\nabla V_0(\theta_i)$
respectively as $N_i\to\infty$, so the lower bounds on $N_i$ in
Proposition~\ref{prop:safety-guarantees} are finite as $N_i\to\infty$.
Therefore, there exists a sufficiently large $N_i$ for which the
conditions on $N_i$ in Proposition~\ref{prop:safety-guarantees} hold.
Next we state a result that provides safety guarantees over a finite time horizon.
Its proof follows from Proposition~\ref{prop:safety-guarantees} and 
Fréchet's inequality~\citep{MF:35}.
\begin{corollary}\longthmtitle{Safety guarantees over a finite time horizon}\label{cor:safety-finite-time-horizon}
    Suppose that Assumptions~\ref{as:boundedness-reward} 
    and~\ref{as:differentiability-lipschitzness-policy} hold.
    Assume also that for each $i\in[H]$, Slater's condition for~\eqref{eq:p-definition} and
    ~\eqref{eq:hat-p-definition} hold, $\hat{A}_{\theta_i}>0$, $A_{\theta_i}>0$, 
    and either the assumptions in 
    Proposition~\ref{prop:safety-guarantees}~\ref{it:safety-guarantees-first}
    or Proposition~\ref{prop:safety-guarantees}~\ref{it:safety-guarantees-second} hold.
    Then, under the iterates of Algorithm~\ref{alg:anytime-safe-rl},
    $\mathbb{P}\Big( \bigcap\limits_{i=1}^{H+1} \{ V_1(\theta_i) \leq 0 \}  \Big) \geq 1-2H\delta$.
\end{corollary}

Next we show the convergence of Algorithm~\ref{alg:anytime-safe-rl} to
a neighborhood of a KKT point
of~\eqref{eq:safe-RL-problem-parameterized-policies}.  Based on
Lemma~\ref{lem:p-properties}~\ref{it:kkt-points}, given a tolerance
$\epsilon^*>0$, we seek to find $\theta\in\real^d$ such that
$V_1(\theta)\leq0$ and $\norm{p(\theta)-\theta} \leq \epsilon^*$. As
shown in~\citep[Theorem 4.1]{RA-JMM-BFS:10}, by sufficiently
decreasing $\epsilon^*$, a vector $\theta\in\real^d$ satisfying
$\norm{p(\theta)-\theta} \leq \epsilon^*$ can be made arbitrarily
close to a KKT point.

\begin{proposition}\longthmtitle{Convergence
  guarantees of RL-SFG}\label{prop:convergence-in-expectation}
Suppose that Assumptions~\ref{as:boundedness-reward}
and~\ref{as:differentiability-lipschitzness-policy} hold.
Suppose there exist positive constants $\widehat{\eta_A}$, $\eta_A$
such that $\hat{A}_{\theta_i} \geq \widehat{\eta_A}$ and
$A_{\theta_i} \geq \eta_A$.  Given $\epsilon^*>0$, there exist
constants\footnote{The constants $M_p$, $\bar{M}_p$, and $K_p$ are
  defined as follows.  First, let $M_q=\sqrt{d}\bar{\sigma}_q$,
  $q\in\{ 0, 1\}$,
  $M_A = ( M_1 )^2 + 2\alpha \tilde{\sigma}_1, M_B = 2M_A, M_C =
  2M_{1} M_{0} + 2\alpha M_{1}$,
  $M_{\Delta} = 4 ( M_{1} + M_{0} )^2 M_A$.  Further define
  $\widehat{\eta_B} = 2 \widehat{\eta_A}$,
  $M_u = \frac{ M_B + \sqrt{ M_{\Delta} } }{ 2\eta_A }$,
  $K_A = 2 M_1$, $K_B = 4 M_1$, $K_C = 2 M_1 + 4 M_0$,
  $K_{\Delta} = \frac{ K_B + 2 M_B K_B + 4K_A M_C + 4 M_A K_C }{
    \widehat{\eta_{\Delta}} }$,
  $K_u = \max\Big\{ \frac{ K_A (M_B \! + \! \sqrt{ M_{\Delta}} )
  }{2\eta_A \widehat{\eta_A} } + \frac{ K_{\Delta} + K_B }{2\eta_A},
  \frac{2 K_C}{\eta_B}, \frac{2 K_C}{\widehat{\eta_B}} \Big\}$.
  Then, we define $M_p = h( M_0 + M_u M_1 )$,
  $\bar{M}_p = \frac{2 M_p}{\frac{1}{h}-\frac{L_0}{2}}$, and
  $K_p = 1+K_u\bar{\sigma}_1+M_u+K_u$.  }  $M_p, \bar{M}_p$ and
$K_p$ such that if
\vspace*{-1ex}
\begin{enumerate}
  \setlength{\itemsep}{0pt}
\item $\epsilon>0$ is such that
  $\sqrt{ \bar{M}_p \epsilon } + h K_p\epsilon \leq \epsilon^*$;
\item $k\in\mathbb{Z}_{>0}$ is such that
  $k\geq\frac{2\tilde{\sigma}_0}{M_p\epsilon}$;
\item for each $i\in[k]$, Slater's condition holds
  for~\eqref{eq:p-definition} and~\eqref{eq:hat-p-definition} with
  $\theta=\theta_i$;
\item $\widehat{V_1}(\theta_i)\leq0$;
\item
  $N_i > \max\{ \frac{d\bar{\sigma}_1^2}{\epsilon},
  \frac{d\bar{\sigma}_0^2}{\epsilon}, \frac{ \tilde{\sigma}_1^2
  }{\epsilon} \}$ for all $i\in[k]$ (with $\bar{\sigma}_0$,
  $\bar{\sigma}_1$, and $\tilde{\sigma}_1$ as defined in
  Lemma~\ref{lem:estimator-statistical-properties});
\end{enumerate}
\vspace*{-1ex}
then, there exists $j\in[k]$ such that
$\mathbb{E}[\norm{\hat{p}(\theta_j)-\theta_j}] \leq \sqrt{ \bar{M}_p \epsilon }$. 
Moreover, $\norm{p(\theta_j)-\theta_j} \leq \epsilon^*$.
\end{proposition}
\smallskip
\vspace{-0.2cm}
\begin{proof}
  First, since $\nabla V_0$ has Lipschitz constant $L_0$,
  by~\citep[Lemma 1.2.3]{YN:18} we have
  $V_0(\hat{p}(\theta_i)) \! \leq \! V_0(\theta_i) \! + \! \nabla
  V_0(\theta_i)^\top ( \hat{p}(\theta_i)-\theta_i ) \! + \!
  \frac{L_0}{2}\norm{\hat{p}(\theta_i)-\theta_i}^2$.  By the
  optimality conditions of~\eqref{eq:hat-p-definition}, it follows
  that
  $(y-\hat{p}(\theta_i))^\top \widehat{\nabla V_0}(\theta_i) +
  \frac{1}{h}(\hat{p}(\theta_i)-\theta_i) \geq 0$ for all
  $y\in\real^d$ that are feasible
  for~\eqref{eq:hat-p-definition}~\citep[Theorem 12.3]{JN-SW:06}.
  Since $\theta_i$ is feasible for~\eqref{eq:hat-p-definition}
  (because $\widehat{V_1}(\theta_i)\leq0$), this implies
  \vspace{-0.1cm}
  \begin{align}\label{eq:v0-bound}
    V_0(\hat{p}(\theta_i))
    &\leq V_0(\theta_i) + (\nabla
      V_0(\theta_i)-\widehat{\nabla
      V_0}(\theta_i) )^\top
      (\hat{p}(\theta_i)-\theta_i)-(\frac{1}{h}-\frac{L_0}{2})\norm{\hat{p}(\theta_i)-\theta_i}^2. 
  \end{align}
  \vspace{-0.2cm}
  After adding up~\eqref{eq:v0-bound} for $i\in[k]$ and rearranging terms we get
  \begin{align}\label{eq:cumulative-p-thetai-0}
    \sum_{i=1}^k \norm{\hat{p}(\theta_i)-\theta_i}^2 = \frac{ V_0(\theta_1)-V_0(\hat{p}(\theta_{k})) }{ \frac{1}{h}-\frac{L_0}{2} }  
    +\sum_{i=1}^k \frac{(\nabla V_0(\theta_i)-\widehat{\nabla
    V_0}(\theta_i) )^\top (\hat{p}(\theta_i)-\theta_i ) }{
    \frac{1}{h}-\frac{L_0}{2} }. 
  \end{align}
  \vspace{-0.3cm}
  On the other hand, by~\cite[Theorem 1.6.2]{RD:10}, for any
  $q\in\{ 0, 1 \}$ and $\theta\in\real^d$ we have
  \begin{align*}
    \mathbb{E}\Bigl[\norm{ \nabla V_q(\theta)\!-\!\widehat{\nabla
    V_q}(\theta)}\Bigr] \! \leq \! \sqrt{ \mathbb{E}\Bigl[\norm{
    \nabla V_q(\theta)\!-\!\widehat{\nabla V_q}(\theta)}^2\Bigr] }, \  
    \mathbb{E}[|V_q(\theta)\!-\!\widehat{V_q}(\theta)|] \! \leq \! \sqrt{ \mathbb{E}[|V_q(\theta) \! - \! 
    \widehat{V_q}(\theta)|^2] }.
  \end{align*}
  Now we note that by
  Lemma~\ref{lem:estimator-statistical-properties},
  $N_i > \max\{ \frac{d\bar{\sigma}_1^2}{\epsilon},
  \frac{d\bar{\sigma}_0^2}{\epsilon}, \frac{ \tilde{\sigma}_1^2
  }{\epsilon} \}$ implies that
  $\Var(\widehat{V_1}(\theta_i)) \leq \epsilon^2$, and
  $\Var( (\widehat{\nabla V_q}(\theta_i))_j ) \leq
  \frac{\epsilon^2}{d} \ \forall j\in[d]$ and $q\in\{ 0, 1 \}$.  Since
  for each $i\in[d]$, the estimate $\widehat{V_1}(\theta_i)$ is
  unbiased, the fact that
  $\Var(\widehat{V_1}(\theta_i)) \leq \epsilon^2$ implies that
  $\mathbb{E}[|V_1(\theta_i)-\widehat{V_1}(\theta_i)|]\leq \epsilon$.
  Moreover, since the estimates $\widehat{\nabla V_q(\theta_i)}$ are
  unbiased for $q\in\{0,1\}$ and $i\in[k]$,
  $\mathbb{E} \biggl[ \norm{ \nabla V_q(\theta_i)-\widehat{\nabla
      V_q}(\theta_i)}^2 \biggr] = \sum_{j=1}^d \Var( (\widehat{\nabla
    V_q}(\theta_i))_j )$ for $q\in\{0,1\}$ and $i\in[k]$.  Now, the
  fact that
  $\Var( (\widehat{\nabla V_q}(\theta_i))_j ) \leq
  \frac{\epsilon^2}{d}$ for all $j\in[d]$ and $q\in\{0,1\}$ implies
  that
  $\mathbb{E}\Bigl[\norm{ \nabla V_q(\theta_i)-\widehat{\nabla
      V_q}(\theta_i)}\Bigr] \leq \epsilon$ for $q\in\{0,1\}$ and
  $i\in[k]$.  Moreover, as shown in the proof of
  Lemma~\ref{lem:estimator-statistical-properties}~\ref{it:gradient-estimator}
  $|(\widehat{\nabla V_q}(\theta_i))_j| \leq \bar{\sigma}_q$ for all
  $j\in[d]$, and since
  $(\nabla V_q(\theta_i))_j = \mathbb{E}[ (\widehat{\nabla V_q}(\theta_i))_j ]$, 
  it also follows that
  $|(\nabla V_q(\theta_i))_j| \leq \bar{\sigma}_q$. This implies that
  $\norm{\widehat{\nabla V_q}(\theta)} \leq M_q$ and
  $\norm{\nabla V_q(\theta)} \leq M_q$ for all $\theta\in\real^d$ and
  $q\in\{0,1\}$.  Similarly, as shown in
  Lemma~\ref{lem:estimator-statistical-properties}~\ref{it:value-function-estimator},
  $|\widehat{V_1}(\theta)|\leq \tilde{\sigma}_1$ for all
  $\theta\in\real^d$, and since
  $V_1(\theta) = \mathbb{E}[ \widehat{V_1}(\theta) ]$ it also follows
  that $|V_1(\theta)| \leq \tilde{\sigma}_1$ for all
  $\theta\in\real^d$.  Now, since the assumptions of
  Lemma~\ref{lem:closed-form-expression-next-iterate} hold, the
  expressions for $\hat{p}$ and $p$ provided in
  Lemma~\ref{lem:closed-form-expression-next-iterate} are satisfied,
  and for all $i\in[k]$ we have that $A_{\theta_i}\leq M_A$,
  $\hat{A}_{\theta_i}\leq M_A$, $B_{\theta_i}\leq M_B$,
  $\hat{B}_{\theta_i}\leq M_B$, $C_{\theta_i}\leq M_C$,
  $\hat{C}_{\theta_i}\leq M_C$, $\Delta_{\theta_i}\leq M_{\Delta}$,
  $\hat{\Delta}_{\theta_i}\leq M_{\Delta}$,
  $\widehat{\eta_B}\leq B_{\theta_i}$, $u_{\theta_i}\leq M_u$,
  $\norm{ \hat{p}(\theta_i)-\theta_i } \leq M_p$, and
  $\norm{ p(\theta_i)-\theta_i } \leq M_p$.  We also have
  $|A_{\theta_i}-\hat{A}_{\theta_i}| \leq K_A\epsilon$,
  $|B_{\theta_i}-\hat{B}_{\theta_i}| \leq K_B\epsilon$,
  $|C_{\theta_i}-\hat{C}_{\theta_i}| \leq K_C\epsilon$,
  $|\sqrt{\hat{\Delta}_{\theta_i}}-\sqrt{ \Delta_{\theta_i}}| \leq
  K_{\Delta}\epsilon$.  From this,~\eqref{eq:cumulative-p-thetai-0}
  implies that
  $\frac{1}{k}\sum_{i=1}^k
  \mathbb{E}[\norm{\hat{p}(\theta_i)-\theta_i}^2] \! \leq \!  \frac{ M_p
    \epsilon}{ \frac{1}{h} -\frac{L_0}{2} } \! + \!  \frac{1}{k}\Big(
  \frac{ \mathbb{E}[V_0(\theta_1)-V_0(\hat{p}(\theta_{k}))] }{
    \frac{1}{h}-\frac{L_0}{2} } \Big)$.  Since
  $|V_0(\theta)|\leq\tilde{\sigma}_0$ for all $\theta\in\real^d$ (as shown in the proof of 
  Lemma~\ref{lem:estimator-statistical-properties}), by choosing
  $k\geq\frac{2\tilde{\sigma}_0}{M_p\epsilon}$ we have
  $\frac{1}{k}\sum_{i=1}^k
  \mathbb{E}[\norm{\hat{p}(\theta_i)-\theta_i}^2] \leq
  \bar{M}_p \epsilon$, and by taking
  $j=\argmin{i\in[k]}{\mathbb{E}[\norm{\hat{p}(\theta_i)-\theta_i}]}$, we have
  $\mathbb{E}[\norm{\hat{p}(\theta_j)-\theta_j}] < \sqrt{ \bar{M}_p \epsilon }$.
  Now, by the triangle inequality
  $\norm{p(\theta_j)-\theta_j} \leq \norm{\hat{p}(\theta_j)-\theta_j}
  + \norm{\hat{p}(\theta_j)-p(\theta_j)}$.
  The rest of the proof focuses on finding an upper bound on $\mathbb{E}[\norm{p(\theta_j)-\hat{p}(\theta_j)}]$.
  To do so, we use Lemma~\ref{lem:closed-form-expression-next-iterate}.
  We first find an upper bound on $|\hat{u}_{\theta}-u_{\theta}|$.
  Note that if $\hat{C}_{\theta_j}<0$ and $C_{\theta_j}<0$, 
  by leveraging the expressions in Lemma~\ref{lem:closed-form-expression-next-iterate} we have
  $|\hat{u}_{\theta_j}\!-\!u_{\theta_j}| \leq \Bigl( \frac{ K_A (M_B
    \! + \! \sqrt{ M_{\Delta}} ) }{2\eta_A \widehat{\eta_A} } + \frac{
    K_B + K_{\Delta} }{2\eta_A} \Bigr)\epsilon$. On the other hand,
  if $\hat{C}_{\theta_j}\geq0$ and $C_{\theta_j}<0$, since
  $\Delta_{\theta_j}=B_{\theta_j}^2-4A_{\theta_j}C_{\theta_j}$,
  we get $|\hat{u}_{\theta_j}-u_{\theta_j}|\leq \frac{2 K_C \epsilon}{ \eta_B }$.
  Similarly, if $\hat{C}_{\theta_j}<0$ and $C_{\theta_j}\geq0$,
  $|\hat{u}_{\theta_j}-u_{\theta_j}| \leq \frac{2
    K_C}{\widehat{\eta_B}\epsilon}$.  All of this implies that
  $|\hat{u}_{\theta_j}-u_{\theta_j}| \leq K_u \epsilon$.  Now, note
  that
  $\norm{\hat{p}-p} \leq h\Big( \frac{\norm{\widehat{\nabla
        V_0}-\nabla V_0} + |\hat{u}-u|\norm{\widehat{\nabla V_1}} + u
    \norm{\widehat{\nabla V_1}-\nabla V_1}}{1 + \hat{u} } +
  \frac{|\hat{u}-u|}{(1+\hat{u})(1+u)} \Big)$,
  and since $u(\theta_j)\geq0$, $\hat{u}(\theta_j)\geq0$, we have 
  $\norm{\hat{p}(\theta_j)-p(\theta_j)}\leq h K_p \epsilon$,
  and $\mathbb{E}[\norm{\hat{p}(\theta_j)-p(\theta_j)}]\leq h K_p \epsilon$,
  from where it follows that $\mathbb{E}[\norm{ p(\theta_j)-\theta_j }] \leq \epsilon^*$.
  Since $\norm{ p(\theta_j)-\theta_j }$ is not random, $\norm{ p(\theta_j)-\theta_j } \leq \epsilon^*$.
\end{proof}

\vspace{-0.2cm}
Proposition~\ref{prop:convergence-in-expectation} ensures that if both
$k$ and each $N_i$, $i\in[k]$, are sufficiently large, RL-SGF
returns a point arbitrarily close to a KKT point
of~\eqref{eq:safe-RL-problem-parameterized-policies}.
In order to detect the index $j\in[k]$ for which 
$\mathbb{E}[\norm{\hat{p}(\theta_j)-\theta_j}] \leq \sqrt{ \bar{M}_p \epsilon }$,
we can compute different realizations of 
$\hat{p}(\theta_i)$ for every $i\in[k]$, compute an empirical estimate of 
$\mathbb{E}[\norm{\hat{p}(\theta_i)-\theta_i}]$, and check whether
$\mathbb{E}[\norm{\hat{p}(\theta_i)-\theta_i}] \leq \sqrt{ \bar{M}_p \epsilon }$.
We note that the issue of identifying the index for which the convergence criterion is satisfied also 
appears in the convergence results of policy gradient 
(cf.~\citep[Theorem 4.3]{KZ-AK-HZ-TB:20}).

%
%

\vspace{-0.2cm}

\section{Simulations}\label{sec:simulations}
%




%
%
In this section we test RL-SGF in a simple navigation 2D example.

\noindent\textbf{Environment:} We consider a continuous state and
action space
environment. We work with two different dynamics:

\vspace{-8pt}
\begin{itemize}
  \item \emph{Single integrator}: $s_{t+1} = s_t + 0.1 a_t$ with
    $s_t\in\Sc = \real^2$ and $a_t \in \Ac = [-5,5]^2$ for
    $t\in\mathbb{Z}_{>0}$;
  \item \emph{Differential-drive robot}:
    $s_{t+1} = s_t + 0.2\: v_t (\cos\theta_t,\sin\theta_t)$ and
    $\theta_{t+1} = \theta_t + 0.2\:\omega_t$ with
    $s_t = (x_t,\theta_t) \in \Sc = \real^2 \times [-\pi,\pi]$,
    $a_t = (v_t,\omega_t) \in \Ac = [0,5] \times
    [-20\pi/180,20\pi/180]$ for $t\in\mathbb{Z}_{>0}$.
\end{itemize}
\vspace{-8pt}
\noindent For the single integrator, $s_t$ is the position of the
agent at time $t$, and for the differential-drive robot, $x_t$ denotes the agent's
position at time $t$ and $\theta_t$ represents its orientation at time
$t$.  For both dynamics, we refer to the position component of the
state as $s_x \in \real^2$.  Given a target state $x^* = (8, 8)$, we
define $R_0(s,a,s^{\prime}) = \max\{-\norm{s_x-x^*},-10\}$.  We
consider a set of five circular and rectangular obstacles denoted by
$\{ \Oc_j \}_{j=1}^5$ as depicted in red in
Figure~\ref{fig:policy_evolution}, and we let the safe set $\Cc$ be
the subset of $[0,10]\times[0,10]\subset\real^2$ not covered by the
obstacles.  We take $T=50$ as the time horizon and $\gamma=0.98$. To
define $R_1(s,a,s^{\prime})$, we let $d_{\min}(s)$ be the minimum
distance between $s_x$ and any obstacle and set
$R_1(s,a,s^{\prime}) = \beta \left(e^{-d_{\min}(s)}-1 \right)$ if
$s_x\in\Cc$ and $R_1(s,a,s^{\prime}) = 1-\beta$ if $s_x\notin\Cc$,
with $\beta>0$ being a design parameter.  This choice of $R_1$ is
inspired by~\cite{SP-MCF-LFOC-AR:23}, and in the limit $\gamma\to 1$
guarantees that trajectories generated by policies satisfying the
associated safety
constraint~\eqref{eq:safe-RL-problem-parametrized-policies-constraint}
do not collide with obstacles with probability greater than
$1-\beta T$. Since $R_0$ and $R_1$ are bounded,
Assumption~\ref{as:boundedness-reward} holds.

\noindent\textbf{Policy:}
We consider a truncated Gaussian policy
class$\pi_\theta(a|s) = \bar{C}
e^{-\frac{1}{2}(a-\mu_\theta(s))^\top\Sigma^{-1}(a-\mu_\theta(s))}$,
if $a\in\Ac$ and $\pi_{\theta}(a|s) = 0$ if $a\notin\Ac$, where
$\Sigma=0.5 \mathbf{I}_2$, and $\bar{C}$ is a normalization
constant. The mean function is defined as
$\mu_\theta(s) = \sum_{i=0}^{N_c} \tanh(\theta_i) \exp\big( -
\frac{\norm{s_x-c_i}^2}{2\sigma^2}\big)$, where $\tanh$ is applied
componentwise, $\theta_i\in\real^2$ are the trainable parameters, and
$\{ c_i \}_{i=1}^{N_c}$ is a set of points in $\Ac$.  For the
single-integrator dynamics, the set $\{c_i\}$ is constructed by
dividing the safe set $\Cc$ into a grid with 20 divisions per
dimension, resulting in $N_c = 400$ points. For the differential-drive
dynamics, an additional dimension with 10 divisions over the range
$[-\pi,\pi]$ is added, leading to $N_c = 4000$ points.  This policy
parameterization satisfies the required properties, i.e.,
$\pi_{\theta}$ satisfies
Assumption~\ref{as:differentiability-lipschitzness-policy}. For the
single integrator, the initial policy parameter $\theta_1$ is
constructed as follows.  Let $\rho, f_{\text{max}}\in \real_{>0}$,
$q_j$ be the center of obstacle $\Oc_j$ and
$v_j^i = (c_i-q_j)/\norm{c_i-q_j}$.  Now, for $i\in[400]$ and
$j\in[5]$, we let
$Q_j^i = f_{\text{max}}(1 - d(c_i,\Oc_j)/\rho) v_j^i$ if
$d(c_i,\Oc_j)<\rho$ and $Q_j^i=0$ otherwise and set
$(\theta_1)_i = \sum_{j=1}^5 Q_j^i$.  As shown in
Figure~\ref{fig:policy_evolution} (left), the resulting mean
$\mu_{\theta_1}$ of $\pi_{\theta_1}$ promotes actions that point away
from the obstacles, yielding a safe initial policy (it can be checked
that for a sufficiently large $N_1$, $\widehat{V_1}(\theta_1)<0$).
For the differential-drive dynamics and due the difficulty of
obtaining an initial safe policy (other than the trivial one where
zero inputs are taken at any state),
%
%
we initizalize parameters $\theta_1$ randomly. This typically results
in an initial policy which is unsafe, and we use this to illustrate
the ability of RL-SGF to recover.

\begin{figure}[htb]
  \centering
  \includegraphics[width=1\textwidth]{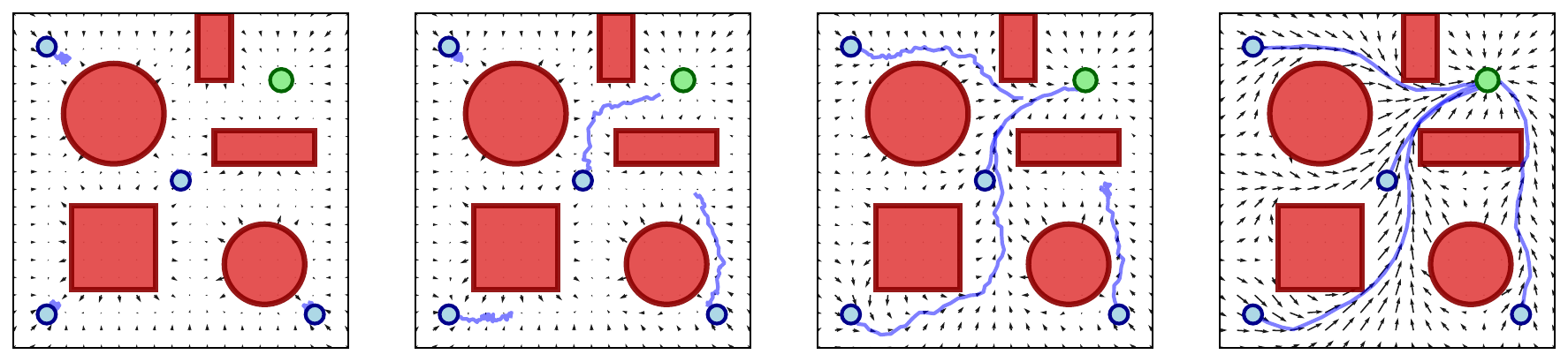}
  \caption{Evolution of the policy $\pi_{\theta}$ generated by RL-SGF
    at diferent stages of the learning process for the
    single-integrator dynamics. The arrows indicate the mean
    $\mu_{\theta}$, the target state $x^*=(8,8)$ is marked in green
    and four trajectories starting at $(1,1),(5,5),(9,1)$ and $(1,9)$
    are plotted in blue. Obstacles are depicted in red.  From left to
    right: initial policy and after 200, 500, and 1500 iterations.}
    \label{fig:policy_evolution}
  \end{figure}

\noindent\textbf{Training:} For the single-integrator case, we
take $N = 100$, $\alpha = 1$, $\beta=0.01$, $h=0.5$ and $K = 1500$. We
compare its performance with the primal-dual method in~\cite[Algorithm
2]{SP-MCF-LFOC-AR:23} with $\eta_\theta=\eta_\lambda=0.001$ as the
primal and dual step sizes and Constrained Policy Optimization
(CPO)~\cite{JA-DH-AT-PA:17}, where we take $\delta = 0.15$ and
$H=\mathbf{I}_d$. To make a fair comparison, we use the same estimates
introduced in Lemma~\ref{lem:estimator-statistical-properties} for all
algorithms, setting $b(s)=0$. For the differential-drive robot case,
we use $N=200$, $\alpha = 9$, $\beta=0.05$, $h=0.1$ and $K=4000$,
keeping all other parameters unchanged.
\begin{figure}[htb]
  \centering
  \includegraphics[width=1\textwidth]{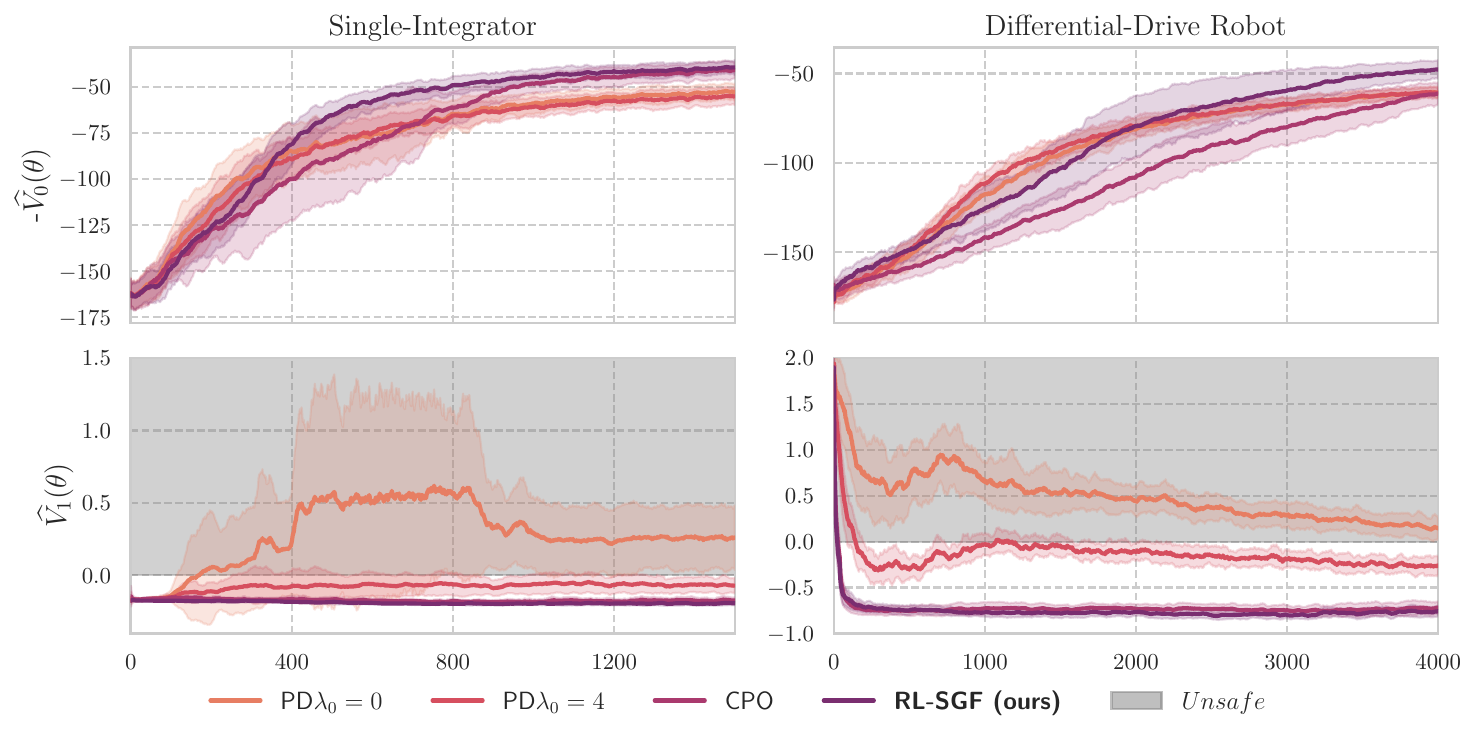}
  \caption{Comparison between RL-SGF, primal-dual approaches (PD), and
    Constrained Policy Optimization (CPO).  Evolution of the average
    return $\widehat{V_0}(\theta)$ and safety measure
    $\widehat{V_1}(\theta)$ for single-integrator (left) and
    differential-drive (right) dynamics. The initial dual variable is
    denoted $\lambda_0$.  Shaded areas represent 95\% confidence
    intervals over 5 runs. The unsafe region
    ($\widehat{V_1}(\theta)>0$) is in gray.}
  \vspace*{-2ex}
    \label{fig:PD-RF-SGF-comparison}
\end{figure}

\noindent\textbf{Results:} Figure~\ref{fig:policy_evolution} shows four
different snapshots describing the policies obtained during the
training process of RL-SGF. Qualitatively, we observe how the vector
field described by the mean of the Gaussian policy always points away
from the obstacles, suggesting that all the policies are safe. We also
observe how policies trained with a larger number of iterations result
in trajectories that make faster progress towards the target point, as
expected.  Figure~\ref{fig:PD-RF-SGF-comparison} compares RL-SGF with
the primal-dual method in~\cite[Algorithm 2]{SP-MCF-LFOC-AR:23} and
CPO~\citep{JA-DH-AT-PA:17} using the same estimators. The initial dual
variable $\lambda_0$ is set to $0$ and $4$.  In both cases, the
primal-dual method results in significant constraint violations,
whereas RL-SGF and CPO maintain safety during training. In
Figure~\ref{fig:PD-RF-SGF-comparison} (right), notice also how RL-SGF
and CPO recover from initial unsafety and reach a safe policy much
faster than the primal-dual algorithm.  Simulations also show that,
for the single integrator, the estimated cumulative reward
$\widehat{V_0(\theta)}$ for RL-SGF becomes larger than that of the
primal-dual algorithm after 300 steps and outperforms CPO throughout
the entire training process.  This is even more apparent in the
differential-drive robot case.  Table~\ref{tab:results} presents the
average performance over the last 100 training steps and the
percentage of safe policies (i.e., $\widehat{V_1}(\theta)\le0$)
observed during the training process across different
algorithms. RL-SGF displays the highest safety performance for both
dynamics, while outperforming other approaches in terms of average
return.


\begin{table}[htb]
  \centering
  \begin{tabular}{l|cccc}
    \textbf{Dynamics}
    & PD $\lambda_0{=}0$ & PD $\lambda_0{=}4$
    &
      CPO & \textbf{RL-SGF (ours)} \\[0.3em] 
    \hline
    Single-integrator &
    -52.99 (26.68\%) &
    -55.53 (85.37\%) &
    -41.33 (99.88\%) &
    \textbf{-39.95 (99.96\%)} \\
    
    Differential-drive \!\!&
    -60.54 (2.17\%) &
    -60.66 (83.28\%) &
    -61.52 (99.72\%) &
    \textbf{-48.01 (99.83\%)} \\
  \end{tabular}
  \vspace*{-1ex}
  \caption{Average performance (\( -\widehat{V_0}(\theta) \)) over the
    last 100 training steps. In parentheses, we report the percentage
    of safe policies (i.e., $\widehat{V_1}(\theta) \leq 0$). Results
    are averaged over 5 random seeds.}\label{tab:results}
\end{table}

Finally, we study the effect on the performance of RL-SGF of
the number of episodes used in the estimates.  For simplicity, we
assume that the number of episodes at each iteration is the same,
i.e., $N_i = N$ for all $i\in[k]$.  Figure~\ref{fig:N-comparison}
shows that, by increasing~$N$, we obtain faster convergence and less
variance on the estimators while reducing the number of unsafe
policies, as expected. In addition, all tested values of $N$ lead to
safe policies during the training process, with minimal constraint
violations even for $N=10$ (less than $0.4\%$). This could be an
indication that, for this environment, the bounds in
Proposition~\ref{prop:safety-guarantees} are conservative. In the case
of $N=400$, RL-SGF achieves zero constraint violation during training.



\begin{figure}[htb]
  \centering
  \includegraphics[width=1\textwidth]{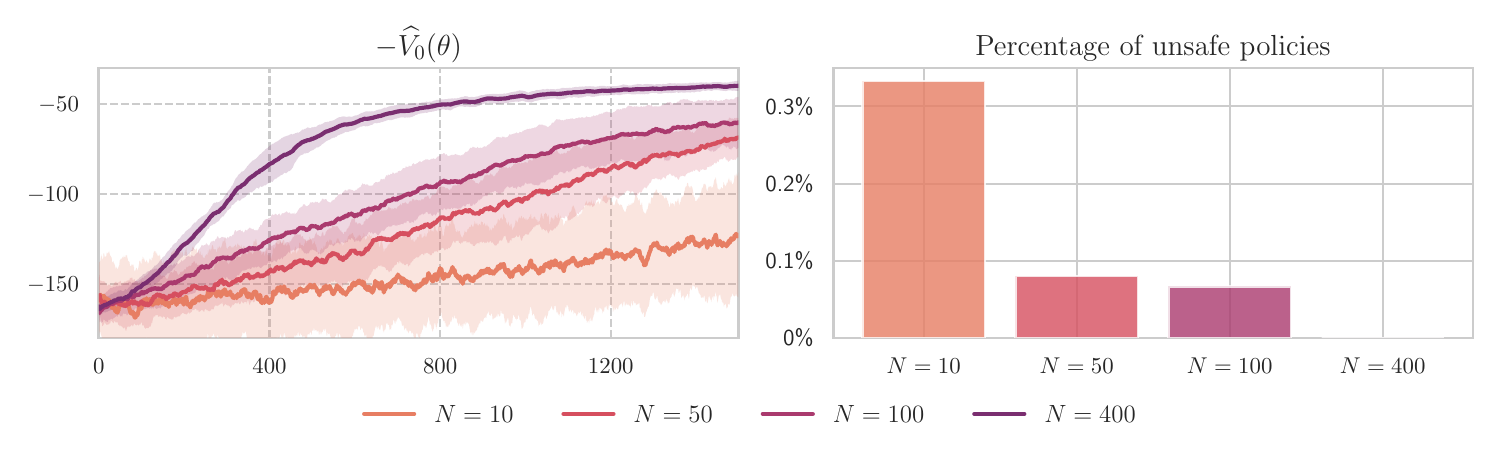}
  \vspace*{-2ex}
  \caption{Illustration of the performance of RL-SGF as a function of
    the number of episodes $N$ used in the estimates of the value
    functions and their gradients. Left plot shows the evolution of
    the average return $\widehat{V_0}(\theta)$ and right plot shows
    the safety measure $\widehat{V_1}(\theta)$ during training. Shaded
    areas are 95 $\%$ confidence intervals over 5 runs.}
    \label{fig:N-comparison}
\end{figure}

\section{Conclusions}

We have introduced RL-SGF, a constrained RL algorithm with anytime
guarantees. At every iteration, RL-SGF uses episodic data to construct
estimates of the value functions and their gradients associated with
the objective and safety constraints. By
deriving appropriate statistical properties of such estimates, we
show that if the number of episodes at each iteration is
sufficiently large, the algorithm returns a safe policy in the next
iteration with arbitrarily high probability, and
converges to a KKT point.  We have illustrated the performance of
RL-SGF in a simple 2D navigation example.  
As future work, we plan to improve the sample-complexity guarantees by
using off-policy methods, explore schemes that adaptively tune the
algorithm parameters to improve convergence, and extend this approach
to other types of safety constraints, such as probabilistic or
conditional value-at-risk.  We also plan to apply this
algorithm to more complex safety-critical systems.

\vspace*{-1ex}
\section*{Acknowledgments}
This work was partially supported by ONR Award N00014-23-1-2353.
\bibliography{../bib/alias,../bib/Main,../bib/Main-add,../bib/FB,../bib/JC}

\appendix
\section{Lipschitz constants}

\begin{lemma}\longthmtitle{Lipschitzness of gradient of value functions}\label{lem:lipschitzness-value-functions}
  Suppose that Assumptions~\ref{as:boundedness-reward}
  and~\ref{as:differentiability-lipschitzness-policy} hold. Let $q\in\{ 0, 1\}$. 
  Then, $\nabla V_q$ is Lipschitz with constant 
  \begin{align*}
      B_q L \Big( \frac{1-\gamma^T}{1-\gamma} \Big)^2 + 
      2 B_q\tilde{B}^2 \gamma \frac{ 1-(T+1)\gamma^T + T\gamma^{T+1} }{ (1-\gamma)^2 } +
      B_q \tilde{B}^2 \Big( \frac{1-\gamma^T}{1-\gamma} \Big)^2.
  \end{align*}
  \smallskip
\end{lemma}
\begin{proof}
  By the Policy Gradient Theorem~\citep[Section 13.2]{RSS-AGB:18}, for any $\theta_1\in\real^d$,
  \begin{align*}
      \nabla V_q(\theta) \! = \! \sum_{t=0}^T \! \sum_{\tau=0}^{T} \int\limits_{\Ic} \! \!
      \gamma^{t+\tau} R_q(s_{t+\tau}, \! a_{t+\tau} , \! s_{t+\tau+1}) \nabla \! \log\pi_{\theta}(a_t | s_t)p_{\theta} d\sigma,
  \end{align*}
  where $\Ic = \Sc^{T+1}\times\Ac^{T+1}$, $d\sigma = ds_0 ds_1 \hdots ds_T da_0 da_1 \hdots da_T$, and 
  \begin{align*}
      p_{\theta} = \Bigg( \prod_{k=0}^{t+\tau} P(s_{k+1},s_k,a_k) \Bigg) \Bigg( \prod_{k=0}^{t+\tau}\pi_{\theta}(a_k|s_k) \Bigg) \eta(s_0).
  \end{align*}
  Now, by following the same steps as in the proof of~\citep[Lemma 3.2]{KZ-AK-HZ-TB:20}, 
  it follows that 
  \begin{align*}
      \norm{\nabla V_q(\theta_1)-\nabla V_q(\theta_2)} &\leq \sum_{t=0}^T \sum_{\tau=0}^T 
      \gamma^{t+\tau} B_q L\norm{\theta_1-\theta_2} + 
      \sum_{t=0}^T \sum_{\tau=0}^T \gamma^{t+\tau} B_q \tilde{B}^2 (t+\tau+1)\norm{\theta_1-\theta_2}
  \end{align*}
  Now, by using the formulas 
  \begin{align*}
      &\sum_{t=0}^T \gamma^t = \frac{1-\gamma^T}{1-\gamma}, \quad 
      \sum_{t=0}^T t \gamma^t = \gamma \frac{ 1-(T+1)\gamma^T + T\gamma^{T+1} }{ (1-\gamma)^2 },
  \end{align*}
  we get 
  \begin{align*}
      \norm{\nabla V_q(\theta_1)-\nabla V_q(\theta_2)} &\leq B_q L \norm{\theta_1-\theta_2} \Big( \frac{1-\gamma^T}{1-\gamma} \Big)^2 \\
      & + 2 B_q\tilde{B}^2 \norm{\theta_1-\theta_2} \gamma \frac{ 1-(T+1)\gamma^T + T\gamma^{T+1} }{ (1-\gamma)^2 } \\
      & + B_q \tilde{B}^2 \norm{\theta_1-\theta_2} \Big( \frac{1-\gamma^T}{1-\gamma} \Big)^2,
  \end{align*}
  from where the result follows.
\end{proof}

\end{document}